\documentclass[journal,twoside, web]{ieeecolor}    
\usepackage{pgfplots}
\usepackage{amsmath,amssymb,amsfonts}
\usepackage{tikz}
\usepackage{float}
\usepackage{cite}
\usepackage{graphicx}
\graphicspath{{images/}}
\usepackage{lcsys}
\pgfplotsset{compat=1.18}

\usepackage{algorithm}
\usepackage{algpseudocode}

\newtheorem{theorem}{Theorem}[section]
\newtheorem{lemma}[theorem]{Lemma}
\newtheorem{proposition}[theorem]{Proposition}

\newtheorem{definition}[theorem]{Definition}

\newtheorem{assumption}[theorem]{Assumption}

\begin{document}

\def\BibTeX{{\rm B\kern-.05em{\sc i\kern-.025em b}\kern-.08em
    T\kern-.1667em\lower.7ex\hbox{E}\kern-.125emX}}
\markboth{\journalname, VOL. XX, NO. XX, XXXX 2017}
{Author \MakeLowercase{\textit{et al.}}: Preparation of Papers for IEEE Control Systems Letters (August 2022)}

\title{Characterization and Computation of Feedback Nash Equilibria in Scalar Discounted N-Player Linear Quadratic Games}

\author{Chiara Cavalagli$^{a}$, Alberto Bemporad$^{a}$, Mario Zanon$^{a}$ 
\thanks{$^{a}$ Research Unit DYSCO (Dynamical Systems, Control, and Optimization) at IMT School for Advanced Studies Lucca.}
\thanks{This work was funded by the European Union (ERC Advanced Research Grant COMPACT, No. $101141351$). Views and opinions expressed are however those of the authors only and do not necessarily reflect those of the European Union or the European Research Council. Neither the European Union nor the granting authority can be held responsible for them.}
}

\maketitle
\thispagestyle{empty}


\begin{abstract}
This paper studies feedback Nash equilibria (FNE) in scalar discounted linear quadratic (LQ) games with $N$ players. By explicitly incorporating the discount factor, we show that finite-cost equilibria may fail to stabilize the original system, motivating a distinction between FNE and stable FNE together with a sufficient stability condition. Based on a parametric characterization of the policies, we propose numerical methods for computing all equilibria. Particular attention is devoted to the symmetric game, where a closed-form expression of the symmetric FNE and conditions for the existence of up to $M\leq2^N-2$ equilibria are derived. Numerical experiments illustrate how equilibrium multiplicity depends on the game configuration and highlight the emergence of finite-cost non-stabilizing equilibria.
\end{abstract}

\begin{IEEEkeywords}
Linear quadratic games, feedback Nash equilibrium, equilibrium computation
\end{IEEEkeywords}



\section{Introduction}
\IEEEPARstart{L}{inear} Quadratic (LQ) discrete-time dynamic games~\cite{noncooperative} model strategic interactions among agents with linear coupled dynamics and quadratic objective functions. Typically constructed under the assumption of adopting a linear feedback law for each player, the solution concept---the Feedback-Nash-Equilibrium (FNE)~\cite{noncooperative}---characterizes the strategy profile as the best outcome that one can achieve with respect to what the other agents are doing. The infinite-horizon formulation, widely studied in the literature, serves as a cornerstone for extensions to constrained~\cite{benenati2025}, finite-horizon~\cite{salizzoni2025hall}, and MPC-based settings~\cite{hall2025solvingmultiparametricgeneralizednash}.  Many numerical methods have been proposed, including Riccati-based iterations~\cite{policy_iteration, model_free_LQ}, gradient-based algorithms~\cite{grad_desce}, non linear least squares~\cite{bemporad2025nashoptpythonlibrary}, and data-driven approaches~\cite{offpolicy_q, model_free_LQ}. 
Despite the amount of literature, the majority of existing methods are primarily designed for the purpose of identifying a single FNE. This stems not only from the complexity of the coupled equations, but also from the difficulty of characterizing the full set of equilibria. Recent efforts have focused on the scalar setting, which, although representing an elementary instance of LQ games, has attracted significant attention due to its rich mathematical structure. In~\cite{Salizzoni_2025_scalar}, the coupled best-response equations are reduced to polynomial systems through Gröbner basis techniques, yielding conditions for uniqueness and multiplicity in the two-player case together with a numerical method for computing all equilibria of relatively small games. In~\cite{feedbackNE}, the authors develop a graphical characterization of scalar $N$-player games through auxiliary functions, deriving conditions for the existence and multiplicity of FNE.\\Our work extends the results of \cite{feedbackNE, Salizzoni_2025_scalar} by addressing the discounted scalar LQ game, where the objective function is weighted through a discount factor $\gamma\in(0,1]$. While discounting is standard in economic models~\cite{engwerda2005lq} and reinforcement learning, its role in the analysis of scalar LQ games has received limited attention. As shown in Section~\ref{sec:general_setting}, introducing the discount factor reveals the existence of finite-cost equilibria that do not necessarily stabilize the original dynamics. This motivates a distinction between FNE and stable FNE, together with a sufficient condition ensuring that all equilibria are stabilizing. Building on the aggregate-value parametrization developed in~\cite{feedbackNE} for the undiscounted game, and in~\cite{engwerda2007algorithms} for the continuous-time counterpart, we derive scalar equations whose roots characterize all equilibria and propose a numerical method for computing all of them. To the best of our knowledge, this is the first complete algorithm for computing all FNE of discounted scalar $N$-player LQ games. We then focus on the symmetric setting, where all players share the same cost parameter. In this case, we derive a closed-form expression for the symmetric FNE together with conditions guaranteeing the existence of up to $M\leq2^N-2$ equilibria. An additional numerical scheme for computing all equilibria is also proposed. Finally, numerical experiments illustrate how equilibrium multiplicity depends on the game parameters and highlight the emergence of finite-cost non-stabilizing equilibria when the sufficient stability condition is violated.



\emph{Notation.}
We denote respectively by $\mathbb{R}, \mathbb{R}_{\geq0}, \mathbb{R}_{>0}$  the set of real numbers, of real positive numbers included zero, and of real positive numbers excluded zero. We denote as $[N]=\{1,\ldots,N\}$ the set of integers between $1$ and $N$, and as $\lfloor N/2\rfloor$ the largest integer smaller than $N/2$. We denote the policy of player $i$ as $k_i$ and of all players except player $i$ as $k_{-i}$. Moreover, we introduce, for every index $i$, the aggregate notation of $S_{-i}= \displaystyle \sum_{j\neq i}k_j=\mathbf{1}^\top k_{-i}$. 

\section{Problem Statement}\label{sec:problem_statement}
Let us consider a scalar non-cooperative discrete-time system, described by an environment state $x_t\in\mathbb{R}$ whose dynamics depend  on the controls $u_{i,t}$ of $N$ players as follows
\begin{equation}\label{eq:LQ}
        x_{t+1}  = a x_t + \displaystyle \sum_{i=1}^{N} u_{i,t}\ ,
\end{equation}
where $a\in\mathbb{R}$ is a fixed parameter. Each player $i$ aims to minimize an infinite-horizon quadratic cost of the form
\begin{equation}\label{eq:cost_j}
J_i(x_t, u_{i,t})  = \displaystyle \displaystyle \sum_{k=t}^{\infty} \gamma^k r_i\left( \sigma_i x_k^2 + u_{i,k}^2 \right),
\end{equation}
with weights $\sigma_i\in\mathbb{R}_{>0}, r_i\in\mathbb{R}_{>0}$ and $\gamma \in(0,1]$ discount factor. The inclusion of the discount factor is a standard technique in reinforcement learning  approaches and economic models for the discount of future losses, as explained in~\cite[Sec.~$3.6$]{engwerda2005lq}. While the framework~\eqref{eq:cost_j} can be rewritten in a standard LQ game~\cite{feedbackNE, Salizzoni_2025_scalar, Nortmann_2023_two_players} via the change of variables $\tilde{a}=a\sqrt{\gamma}$, $\tilde{\sigma_i}=\gamma \sigma_i$ and $\tilde{k_i}=k_i\sqrt{\gamma}$, keeping the discount factor explicit makes it possible to distinguish between stability of the original closed-loop system and convergence of the discounted cost, as discussed shortly. This class of problems~\eqref{eq:LQ}-\eqref{eq:cost_j}, called linear quadratic (LQ) games~\cite{noncooperative}, represents the game-theoretical equivalent of the linear quadratic regulator (LQR).

In this work, we restrict attention to the class of linear state-feedback strategies $u_{i,t} = k_i x_t$, $k_i \in \mathbb{R}$, whose corresponding solution is referred to as feedback Nash equilibrium (FNE). While more general (e.g., nonlinear or open-loop) strategies could be considered, the linear feedback assumption significantly simplifies the analysis; it should therefore be understood as a structural restriction of the present work.

For every fixed vector of policies $k=(k_1, \ldots,k_N)$, the closed loop system evolves as $x_{t+1}=a_\mathrm{cl}(k)x_t$, with $a_\mathrm{cl}(k)=\left(a  + \sum_{i=1}^{N} k_i\right)$. Given an initial state $x_0\in\mathbb{R}$, by substituting the closed-loop dynamics into the objective function~\eqref{eq:cost_j}, the corresponding value function for player $i$ is
\begin{equation}\label{eq:value_function_0}
    V_i(x_0; k_i,k_{-i})=  x_0^2r_i\left( \sigma_i + k^2_i  \right)\displaystyle \sum_{t=0}^{\infty}\left( \gamma a_{\mathrm{cl}}^2(k)\right)^t ,
\end{equation}
which converges if and only if $\left|a_\mathrm{cl}(k)\right|<\frac{1}{\sqrt{\gamma}}$, enlarging the subset of solutions yielding finite cost but possibly non-stabilizing the system. This distinction motivates separating the notions of FNE, and stable FNE, depending on whether the induced closed-loop system \eqref{eq:LQ} is stable.

\begin{definition}[Stable feedback Nash equilibrium] \label{def:stable_fne}
    Given a discounted $N$-player LQ game, the feedback $k^\star$ is a FNE if, for all $x_0$ and all $k_i$: $V^{\star}_i(x_0):=V_i(x_0;k^\star_i,k^\star_{-i})\leq V_i(x_0;k_i,k^\star_{-i})$. Moreover, $k^\star$ is a stable FNE if, in addition,        \begin{equation}\label{eq:stability_oc}
            |a_{\mathrm{cl}}(k^\star)|<1.
        \end{equation}
\end{definition}

From the definition of value function in~\eqref{eq:value_function_0},
a polynomial characterization of the stable FNE is obtained, as formalized next. 
\begin{proposition}[Problem Statement]\label{prop:problem_statement}
Let us consider a discounted $N$-player LQ game, and a vector $k=(k_i, k_{-i})$ such that $\left|a_\mathrm{cl}(k)\right|<\frac{1}{\sqrt{\gamma}}$. Then $k$ is a stable FNE if it solves
\begin{equation}\label{eq:stable_FNE_sys}
\left\{
\begin{array}{l}
f_i(k_i,k_{-i})=0,\ \forall i\in[N],\\
|a_{\mathrm{cl}}(k)|<1
\end{array}
\right.
\end{equation}
where
\begin{equation}
\begin{aligned}\label{eq:best_response}
    f_i (k_i, k_{-i})= &-\gamma k_i^{2} \left(a + S_{-i}\right)
+ \gamma \sigma_i \left(a + S_{-i}\right)\\
&\ - k_i\left(\gamma\left(a + S_{-i}\right)^{2}-\gamma\sigma_i-1\right)=0,
\end{aligned}
\end{equation} 
is the first order condition that characterizes the best response of player $i$.
\end{proposition}
\begin{proof} 
    The value function $V_i$ of the $i$-th player is defined by the geometric series~\eqref{eq:value_function_0}, which, thanks to the assumption $\left|a_\mathrm{cl}(k)\right|<\frac{1}{\sqrt{\gamma}}$ converges to
    \begin{equation}\label{eq:value_function}
        V_i(x_0, k_i, k_{-i})=\frac{x_0^2r_i\left( \sigma_i + k^2_i  \right)}{1-\gamma a_{\mathrm{cl}}^2(k)}.
    \end{equation}
    We can then formulate the optimality conditions for the stable FNE from Definition~\eqref{def:stable_fne} as follows:
    \begin{equation}\label{eq:NEC}
        \frac{\partial V_i}{\partial k_i}=\frac{2 \gamma a_\mathrm{cl}(k)\left(\sigma_i + k_{i}^{2}\right)+ 2 k_{i} \left(1-\gamma a_\mathrm{cl}^2(k) \right)}{\left(1-\gamma a_{\mathrm{cl}}^2(k)\right)^{2}}=0.
    \end{equation}
    By rewriting $a_\mathrm{cl}(k)=a+k_i+S_{-i}$ and by restricting ourselves to stable systems~\eqref{eq:stability_oc}, the denominator is by construction never $0$; by considering only the numerator, we obtain~\eqref{eq:best_response}. Moreover, under $\left|a_\mathrm{cl}(k)\right|<\frac{1}{\sqrt{\gamma}}$,
    condition~\eqref{eq:NEC} is also sufficient for optimality, and therefore uniquely characterizes the best response of player~$i$.
\end{proof}
In the next section, we provide an analysis of the best response system described by~\eqref{eq:stable_FNE_sys} which will be used for the development of a numerical method that seeks all stable FNE.

\section{The General Setting}\label{sec:general_setting}
In this section, we provide a geometrical analysis of the system~\eqref{eq:best_response}, from which we firstly discuss the distinction between stable FNE and FNE, to then provide a parametric characterization of all solutions. Throughout the whole work, we will assume non-zero policies $k_i\neq 0$ for every player and non-zero game-parameter $a\neq 0$. 

\begin{lemma}\label{lemma:factorization_sys}
    The first order condition function $f_i(k_i, k_{-i})=0$ can be factorized as $f_i(k_i, k_{-i})=-\gamma k_i\left(S_{-i}-h_i^-(k_i)\right) \left(S_{-i}-h_i^+(k_i)\right)=0$, where 
\begin{equation}\label{eq:roots}
\begin{aligned}
        \displaystyle h_i^\pm(k_i)=&\frac{\sigma_i-k_i(2a+k_i)}{2k_i} \pm\frac{\sqrt{\gamma^2(\sigma_i+k_i^2)^2+4\gamma k_i^2}}{2\gamma k_i},
\end{aligned}
\end{equation}
are the two roots of $f_i(k_i, k_{-i})=0$ iteratively solved w.r.t. every player policy except the $i$-th one.
\end{lemma}
The first order condition of player $i$ is therefore linear in $S_{-i}$ and nonlinear in $k_i$, a crucial relation for the characterization of the solutions. As proved in earlier works for the undiscounted case~\cite[Lemma~$1$]{feedbackNE}\cite[Thm.~$5$]{Salizzoni_2025_scalar}, only one of the two roots yields stable FNE. This is not always true for the discounted version. As already shown for the discounted two-player game in~\cite[Prop.~$4.2$]{cavalagli2026}, we extend such result for the $N$-player case in the following Proposition.

\begin{proposition}\label{prop:ONLY_ONE_ROOT}
Let $k^\star=(k_i^\star, k_{-i}^\star)$ be a FNE.
\begin{equation}\label{eq:nec_cond}
    k^\star \text{ is a stable FNE}  \Rightarrow\ 
        S_{-i}^\star =h_i^-(k_i^\star)\ \forall i,
\end{equation}
\begin{equation}\label{eq:suff_cond}
    k^\star \text{ is a stable FNE}\Leftarrow\exists\ i:\  
    \begin{cases}
        &S_{-i}^\star =h_i^-(k_i^\star)   \\
        &\sigma_i>\frac{(1-\gamma)^2}{4\gamma^2}  
    \end{cases}.
\end{equation}
\end{proposition}
\begin{proof}
    $(\Rightarrow)$. Consider a stable FNE $(k_i, k_{-i})$. For each player $i$, by Lemma~\ref{lemma:factorization_sys},  either $h_i^-(k_i)=S_{-i}$ or $h_i^+(k_i)=S_{-i}$. Consider the latter and substitute it inside the stability condition $\left|a+k_i+S_{-i}\right|=\left|a+k_i+h_i^+(k_i)\right|=\left|z_i^\mathrm{u}(k_i)\right|$. One can verify that the function $z_i^\mathrm{u}(k_i)$ is unbounded as $|k_i|\to\infty$ and as $k_i\rightarrow 0$. Moreover, its extrema are attained at $\overline{k_i}=\pm\sqrt{\sigma_i}$, where $z_i^\mathrm{u}(\pm\sqrt{\sigma_i})=\pm\left(\sqrt{\sigma_i}+\sqrt{\sigma_i+1/\gamma}\right)$, which are always larger than $1$ in absolute value. Hence, $|z_i^{\mathrm u}(k_i)|>1$ for every admissible $k_i$, contradicting stability. Therefore, the only branch that can yield a stable FNE is $h_i^-(k_i)$. $(\Leftarrow)$. Consider FNE $(k_i, k_{-i})$ such that $S_{-i}=h_i^-(k_i)$ for a given $i$ and let us prove that it is stable. To do that, let us consider the stability condition $\left|a+k_i+S_{-i}\right|<1$ and substitute the hypothesis $\left|a+k_i+h_i^-(k_i)\right|=\left|z_i^\mathrm{s}(k_i)\right|$. The function $z_i^\mathrm{s}(k_i)$ converges to $0$ as $k_i\to0$ and as $|k_i|\to\infty$. Its stationary points are attained at $k_i=\pm\sqrt{\sigma_i}$, with $z_i^\mathrm{s}(\pm\sqrt{\sigma_i})=\mp\left(\sqrt{\sigma_i+1/\gamma}-\sqrt{\sigma_i}\right)$ being strictly inside $(-1, 1)$ if and only if $\sigma_i>\frac{(1-\gamma)^2}{4\gamma^2}$. Hence $|a+k_i+S_{-i}|<1$, and the FNE is stable.
\end{proof}
It is important to note that~\eqref{eq:suff_cond} is sufficient and necessary when $\gamma=1$, as discussed in~\cite{feedbackNE, Salizzoni_2025_scalar}. We will assume from now on that condition~\eqref{eq:suff_cond} is valid, such that every FNE is also stabilizing. Therefore, to simplify the exposition and remain consistent with the standard terminology adopted in the LQ games literature, we will simply refer to stable FNE as FNE.

\begin{proposition}[Optimal problem statement]
    Let us consider a discounted $N$-player LQ game~\eqref{eq:stable_FNE_sys} such that condition~\eqref{eq:suff_cond} is valid, and the vector $k^\star=(k_i^\star, k_{-i}^\star)$. Then $k^\star$ is a FNE if and only if it is a solution to the following system
    \begin{subequations}
    \label{eq:optimal_system}
    \begin{align}
&S^\star =h_i^-(k_i^\star)+k_i^\star, \quad \forall i\in[N], \label{eq:optimal_system_1}\\
    &\left|S^\star+a\right|\leq \left|\sqrt{\bar\sigma+\frac{1}{\gamma}}-\sqrt{\bar\sigma}\right|,\label{eq:optimal_system_2}
\end{align}
\end{subequations}
where $\bar\sigma := \max_i \sigma_i$. 
\end{proposition}
\begin{proof}
    Let us consider the optimal first order condition $S_{-i}^\star=h_i^-(k_i^\star)$ and let us add $k_i^\star$ to both sides of the equation $S_{-i}^\star + k_i^\star =h_i^-(k_i^\star)+k_i^\star$. By defining the RHS of the equation as $F_i(k_i^\star):=h_i^-(k_i^\star)+k_i^\star$, one can observe that $F_i(k_i)$ is a bounded function in $k_i$ with horizontal asymptote $y=-a$. Its extrema are attained for $k_i=\pm\sqrt{\sigma_i}$ such that $F_i(k_i)-a\in\left[\sqrt{\sigma_i}-\sqrt{\sigma_i+\frac{1}{\gamma}}, -\sqrt{\sigma_i}+\sqrt{\sigma_i+\frac{1}{\gamma}}\right]$. All admissible intersections between $F_i(k_i)$  and $y=S^\star$ lie in the same half-space of $\mathbb{R}$ for all $i$, leading to all players having the same sign at the FNE. For simplicity, we restrict the study of the function in the negative half-space, which corresponds to $a>0$. In order for every $F_i(k_i)$ to admit a real intersection with $y=S^\star$, observe that the quantity $\left|\sqrt{\sigma_i+\frac{1}{\gamma}}-\sqrt{\sigma_i}\right|$ is minimized when $\sigma_i=\bar\sigma:=\displaystyle\max_{i}\sigma_i$. Hence, a feasible solution exists if $\left|S^\star+a\right|\leq \left|\sqrt{\bar\sigma+\frac{1}{\gamma}}-\sqrt{\bar\sigma}\right|$, which defines the \emph{feasibility condition}. Finally, $S^\star\neq -a$, since substituting $S^\star=-a$ into~\eqref{eq:optimal_system_1} yields $k_i^\star=0$ as the unique solution.
\end{proof}
It is important to note that for $\gamma=1$, condition~\eqref{eq:optimal_system_1} is equivalent to the closed loop bound given in~\cite[Thm. $2$]{feedbackNE} for the undiscounted problem. Now, let us suppose there exists $S^\star$ such that condition~\eqref{eq:optimal_system_1} holds,
and let us study the expression $h_i^-(k_i)+k_i=S^\star$ which yields a FNE $k^\star=(k_i^\star, k_{-i}^\star)$. Given the presence of square-root terms, a symbolic solving procedure typically requires to square up both sides with the aim to obtain a polynomial expression. 
After few steps one obtains that the squared expression is a quadratic in $S^\star$, i.e., $\gamma k_i^2 \left(S^\star + a \right) + k_i \left( 1 - \left(S^\star+a\right)^2\right) + \gamma\sigma_i \left(S^\star+a\right)=0$, whose solutions are
\begin{equation} \label{eq:FNE_charact}      
k_{i,\pm}^\star(S^\star)=\frac{\gamma (S^\star+a)^2-1}{2\gamma(S^\star+a)}\pm\frac{\sqrt{\Delta^{(i)}(S^\star)}}{2\gamma(S^\star+a)},
\end{equation}
with $\Delta^{(i)}(S^\star):=\left(\gamma(S^\star+a)^2-1\right)^2-4\gamma^2\sigma_i(S^\star+a)^2$. Because they satisfy $k_{i,-}^\star(S^\star) k_{i,+}^\star(S^\star)=\sigma_i$ for every $i$, we will denote these solutions as \emph{hyperbolic}. Thus, for a fixed feasible $S^\star$, there exists a vector of $N$ sign assignments $\epsilon_j=(\epsilon_j(i))_i\in\{-1,1\}^N$ such that
\begin{small}
\begin{equation}\label{eq:canon_S}
\begin{aligned}    S^\star_{\epsilon_j}=\sum_{i=1}^Nk_{i,\epsilon_j(i)}^\star(S^\star) =& N\cdot\frac{\gamma (S^\star+a)^2-1}{2\gamma(S^\star+a)} \\
    &+\frac{1}{2\gamma(S^\star+a)}\sum_{i=1}^N\epsilon_j(i)\sqrt{\Delta^{(i)}(S^\star)}.
\end{aligned}
\end{equation}
\end{small}
By substituting the symbolic expression of $\Delta^{(i)}(S^\star)$ for every $i$, we obtain a one-variable expression which can be solved numerically. By counting every possible combination of sign assignments, there exist at most $\left|\{-1,1\}^N\right|=2^N$ distinct $S^\star$. However, as proven in~\cite{feedbackNE}, there are at most $2^N-1$ FNE, so that at least one combination is always infeasible. Based on this mathematical structure, we developed a numerical method that finds all FNE whose pseudo-code is given in Algorithm~\ref{alg:code}. The resulting scalar nonlinear equations are solved numerically using the \texttt{nsolve} routine of the SymPy library. In the undiscounted case $\gamma=1$, equations~\eqref{eq:FNE_charact} and~\eqref{eq:canon_S}$+a$ reduce, respectively, to the FNE and closed-loop characterization equations derived in \cite[Eq.~$(9)$ and Eq.~$(12)$]{feedbackNE}. Hence, the present formulation recovers the classical undiscounted setting as a special case while extending it to discounted LQ games with $\gamma<1$. In the next section, we focus on the special case $\sigma_i=\sigma\, \ \forall\, i\in[N]$, which we denote as \emph{symmetric setting}, as that will allow us to sharpen our results, and that will help improve the understanding of how multiple FNE arise.
\begin{algorithm}[t]
\caption{FNE-seeking method - LQ Game}
\label{alg:code}
\begin{algorithmic}
\State \textbf{Input: } $a, \gamma, (\sigma_i)_i$
\State \textbf{Output:}  List of FNE
\State Compute $2^N$ symbolic expressions of $\{S_{\epsilon_{j}}\}_j$~\eqref{eq:canon_S}
\For{each combination $\epsilon_j\in\{-1,1\}^N$}
\State Find root $S^\star$ : $S^\star -S^\star_{\epsilon_j}=0$
\If{$S^\star$ is feasible} 
\State Retrieve FNE $\left(k_{i, \epsilon_j(i)}^\star(S^\star)\right)_i$
\EndIf
\EndFor 
\end{algorithmic}
\end{algorithm}

\section{The Symmetric Setting}\label{sec:symmetric_setting}
Let $\sigma_i=\sigma_j$ for every $i, j\in[N]$ and refer to such a setting as symmetric. As already mentioned earlier, condition~\eqref{eq:suff_cond} is assumed to hold, such that every FNE is stabilizing. Similarly to what has been developed for the undiscounted continuous counterpart in \cite[Sec.~$4$]{ENGWERDA2016364}, and as an extension for the discrete-time discounted two-player game in~\cite[Sec.~$4.1$]{cavalagli2026}, the proposed setting yields closed-form solutions together with a hierarchical structure that shows how multiple FNE arise. The optimal system~\eqref{eq:optimal_system} reduces to 
\begin{subequations}
    \begin{align}        
        &S^\star = h^-(k_i)+k_i, \quad \forall i\in[N],\label{eq:optiaml_sys_symm}\\
        &\left|S^\star+a\right|\leq \left|\sqrt{\sigma+\frac{1}{\gamma}}-\sqrt{\sigma}\right|,\label{eq:feasibility_S_symm}
    \end{align}
\end{subequations}
where we drop the subscript in $h^-(k_i)$ as the agents are all equal. As already proved for the differential undiscounted counterpart in \cite[Thm. ~$4.3$]{ENGWERDA2016364}, the symmetric setting always admits a symmetric FNE with a closed-form expression. Let us formalize it in the next Theorem.
\begin{theorem}\label{thm:symm}
The symmetric discounted LQ game always admits a unique symmetric stable FNE $k^\star_\mathrm{s}=(k_\mathrm{s}, \ldots, k_\mathrm{s})\in\mathbb{R}^N$, where
    \begin{equation}\label{eq:symm_FNE}
        \displaystyle k_\mathrm{s}= -\frac{1}{3}\left( S_N \omega + U_N+ \frac{p_N}{S_N \omega}\right),\quad \omega = -\frac{1}{2} - \frac{\sqrt{3}}{2} i
    \end{equation}
    with   
    \begin{align*}
    p_N=& - \frac{3 \left(N \gamma \sigma + a^{2} \gamma - 1\right)}{N^{2} \gamma + N \gamma} + \frac{(2Na+a)^2}{(N^{2}+N)^2}, \\
    q_N=& - \frac{27 a \sigma}{N^{2}+N} \\
    &- \frac{9 (2Na+a)\left(N \gamma \sigma + a^{2} \gamma - 1\right)} {(N^{2}+N)(N^{2}\gamma+N\gamma)} +\frac{2(2Na+a)^3}{(N^{2}+N)^3},
    \end{align*}    
    and 
    \begin{equation*}
        \Delta_N = \sqrt{\,4 p_N^{3} + q_N^{2}\,},\ S_N = \sqrt[3]{\frac{-q_N + \Delta_N}{2}},\ U_N = \frac{2Na+a}{N^{2}+N}.
    \end{equation*}
\end{theorem}
\begin{proof}
    The proof is given in Appendix~\ref{appendix:thm_symm}.
\end{proof}
Let us assume next that $\exists\  i\in[N]:k_i\neq k_j$ for some $j\in[N]$ such that conditions~\eqref{eq:feasibility_S_symm}-\eqref{eq:optiaml_sys_symm} hold. Let us consider the equation $S^\star=h^-(k_i)+k_i$ and, as we did in Section~\ref{sec:general_setting}, square the expression and obtain a quadratic admitting the hyperbolic roots $k_{i, \pm}^\star(S^\star)$ of~\eqref{eq:FNE_charact}. Because each player admits the same cost value $\sigma_i=\sigma$, then the subscript is no longer needed: $k_{i, \pm}^\star(S^\star)=k_\pm^\star(S^\star)$ for every $i$. As a result, there exists an index $p\in[N-1]=\{1,\ldots,N-1\}$ such that
$S^\star$ from~\eqref{eq:canon_S} is the sum of $p$-times $k_+^\star(S^\star)$ and $(N-p)$-times $k_-^\star(S^\star)$, i.e., 
\begin{equation}\label{eq:S_p}
    S_p^\star = p\cdot k_+^\star(S^\star_p) + (N-p)\cdot k_-^\star(S^\star_p),
\end{equation}
where we use the subscript $S_p^\star$ to remark the dependence on $p$. To further simplify notation, let us consider the set of indices of the players that choose policy $k_+^\star(S^\star)$ as $\mathcal{I}^+=\{i\in[N]\ |\ k_i=k_+^\star(S^\star) \}$, and a complementary set of players choosing policy $k_-^\star(S^\star)$ as $\mathcal{I}^-=\{i\in[N]\ |\ k_i=k_-^\star(S^\star) \}$. 
One can observe that equation~\eqref{eq:S_p} depends only on the number $p$ of players selecting $k_+^\star(S_p^\star)$, and not on their specific indices. Therefore, every possible assignment of $p$ players to the set $\mathcal I^+$ and the remaining $N-p$ players to $\mathcal I^-$ yields the same aggregate value $S_p^\star$. Nevertheless, each assignment corresponds to a distinct FNE, since the resulting feedback vector $k^\star$ differs by permutation of its components. As a consequence, every solution $S_p^\star$ generates $\binom{N}{p}$ distinct FNE.


Thus, computing the FNE narrows down to finding $S^\star_p$ as a root of Equation~\eqref{eq:S_p} for every $p\in[N-1]$, which, however, cannot be solved in closed-form. By underlining the fact that the set of all solutions can be found by Algorithm~\ref{alg:code}, we propose an alternative numerical method based on how multiple FNE hierarchically arise. 
While the previous numerical method computes all admissible values of $S^\star$, the alternative approach exploits the hyperbolic structure of the equilibria and directly searches for the pairs $k_\pm^\star(S_p^\star)$ associated with the same aggregate value $S_p^\star$. To that end, let us take the $i$-th equation of~\eqref{eq:optiaml_sys_symm}, consider a global $k_i=k\in\mathbb{R}$, substitute the relation $k_+^\star=k, k_-^\star=\frac{\sigma}{k}$ inside~\eqref{eq:S_p}, and obtain the following 
\begin{equation}
    \eqref{eq:optiaml_sys_symm}\underset{\forall p\in[N-1]}{\iff}\ h^-(k)+k-S_p(k)=\Phi_p(k)=0,
\end{equation}
with 
\begin{equation}\label{eq:phi_k}
\begin{aligned}
     \Phi_p(k)=& \frac{- 2 N \gamma \sigma - 2 a \gamma k + \gamma \left(\sigma + k^{2}\right) }{2 \gamma k}\\
     & +\frac{p \left(2 \gamma \sigma - 2 \gamma k^{2}\right)-\sqrt{\gamma^2\left(k^2+\sigma\right)^2+4\gamma k^2}}{2 \gamma k}. 
\end{aligned}   
\end{equation}
For every $p\in[N-1]$, we numerically solve~\eqref{eq:phi_k} for a coordinate $k$, retrieve its hyperbolic counterpart $\sigma/k$, and generate $\binom{N}{p}$ distinct FNE through permutation of the two policies over the $N$ players. As illustrated in Figure~\ref{fig:phi_p}, numerical evidence suggests that the equation $\Phi_p(k)=0$ typically admits two distinct roots, each defining a different hyperbolic pair $\left(k,\sigma/k\right)$ and therefore a different combinatorial family of equilibria. An exception arises when $N$ is even and $p=N/2$, where the two hyperbolic branches collapse into the same family of equilibria. As the observed behavior cannot currently be established analytically, we state it as an assumption.\begin{assumption}\label{ass:phi_roots}
For every $p\in[N-1]$, the scalar equation $\Phi_p(k)=0$ admits two admissible roots, each associated with a hyperbolic pair $\left(k,\sigma/k\right)$. Moreover, when $N$ is even and $p=N/2$, the two hyperbolic branches generate the same family of equilibria.
\end{assumption}
In the next Theorem, we formalize the resulting hierarchical structure under the presented assumption together with a sufficient and necessary condition, when $N$ is even, and a sufficient condition, when $N$ is odd, for the existence of the maximal number $2^N-2$ of hyperbolic equilibria, which together with the symmetric FNE yields $2^N-1$ solutions. This recovers the multiplicity result of~\cite[Thm.~$2$]{feedbackNE} for the discrete-time undiscounted game and of~\cite[Thm.~$4.3$]{ENGWERDA2016364} for the undiscounted differential game.
\begin{theorem}
Under Assumption~\ref{ass:phi_roots}, the symmetric discounted LQ game admits up to
\begin{equation}\label{eq:cardinality}
    M=
    \displaystyle2\sum_{z=1}^{\hat{p}} \binom{N}{z}
    -\mathbf{1}_{\{N\text{ even, }\hat{p}=N/2\}}
    \binom{N}{N/2}
\end{equation}
hyperbolic FNE, where $\hat{p}:=\displaystyle\max_{\{1,\ldots,\lfloor N/2\rfloor\}}p$ such that
\begin{equation}\label{eq:condition_multiplicity}
    \Phi_{\hat{p}}(\bar{k}_{\hat{p}})\mathrm{sign}(-a)\geq0\quad \text{with }\ \bar{k}_{\hat{p}}:\ \Phi_p'(\bar{k}_{\hat{p}})=0.
\end{equation}

Moreover, the game admits the maximal number $2^N-2$ of hyperbolic FNE if
\begin{equation}\label{eq:N_even}
    |a| \ge \sqrt{\sigma}(N-1)+\sqrt{\sigma+\frac{1}{\gamma}}.
\end{equation}
If $N$ is even, condition~\eqref{eq:N_even} is also necessary.
\end{theorem}
\begin{figure}[t!]
    \centering
    \includegraphics[width=\linewidth]{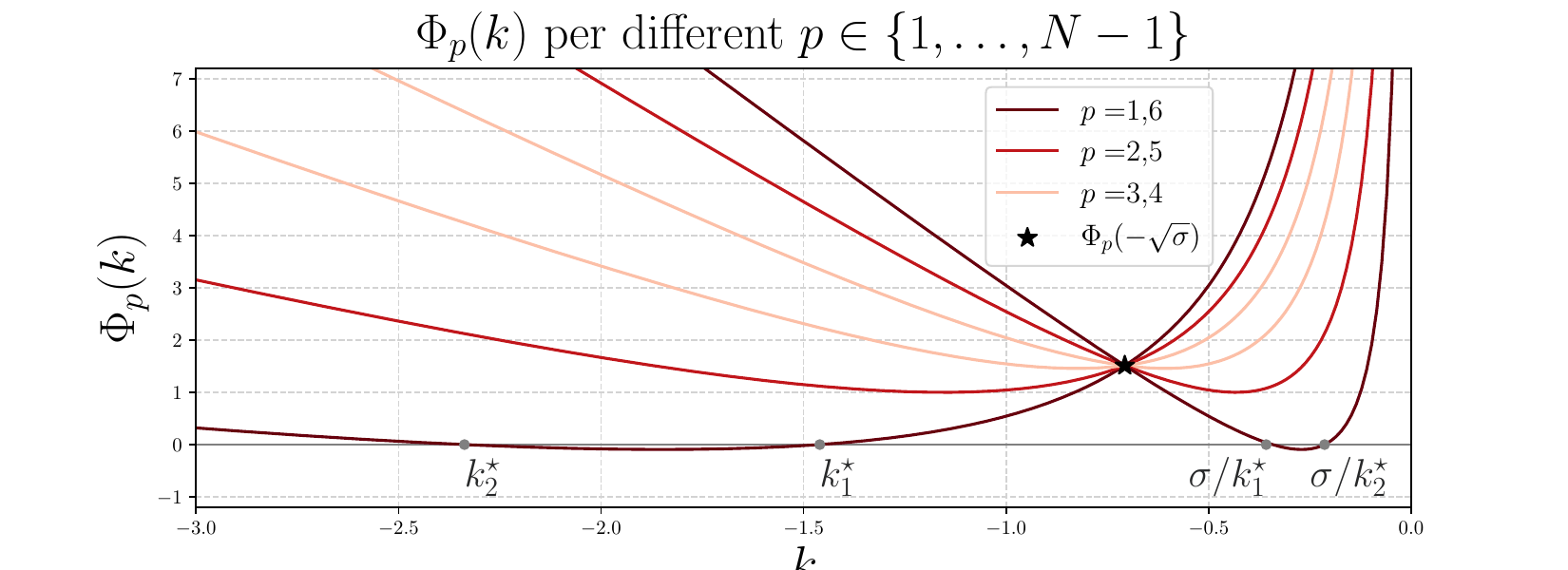}
    \caption{Visualization of $\Phi_p(k)$~\eqref{eq:phi_k} per different values of $p\in[N-1]$. The intersection with the $k$-axis yields a couple of stable policies $(k^\star_{1,p}, k^\star_{2,p})$ to permute with their corresponding hyperbolic value over the $N$-vector. }
    \label{fig:phi_p}
\end{figure}
\begin{proof}
For every $p\in[N-1]$, the function $\Phi_p(k)$ is unbounded, since $\displaystyle\lim_{k\to0}\Phi_p(k)=\pm\infty$ and $\displaystyle\lim_{k\to\pm\infty}\Phi_p(k)=\mp\infty$. 
Moreover, straightforward algebraic manipulations yield, for every $p\in[N-2]$,
\begin{equation}\label{eq:prop_phi}
    \Phi_{p+1}(k)-\Phi_p(k)=\frac{\sigma}{k}-k,
    \quad
    \Phi_p\left(\frac{\sigma}{k}\right)=\Phi_{N-p}(k).
\end{equation}

The first identity shows that the family $\{\Phi_p\}_{p=1}^{N-1}$ is ordered through a hyperbolic shift, while the second establishes a symmetry between the indices $p$ and $N-p$. As a consequence, it is sufficient to study the first $\lfloor N/2\rfloor$ functions. Under Assumption~\ref{ass:phi_roots}, every admissible root of $\Phi_p(k)=0$ generates a hyperbolic pair $(k,\sigma/k)$ and therefore a corresponding combinatorial family of equilibria. Hence, additional hyperbolic FNE arise whenever $\Phi_p(k)$ crosses the horizontal axis. Since $\Phi_p(k)$ is continuous on each half-space and diverges at the boundaries, the latter condition occurs whenever a stationary point $\bar{k}_p$, satisfying $\Phi_p'(\bar{k}_p)=0$, satisfies $\Phi_p(\bar{k}_p)\,\mathrm{sign}(-a)\ge0$. By defining $\hat p$ as the largest index satisfying the mentioned condition, summation over all admissible combinatorial families yields~\eqref{eq:cardinality}. Consider next the maximal multiplicity case. Let $a>0$, since the argument for $a<0$ is symmetric. If $N$ is even, then $\lfloor N/2\rfloor=N/2$, and one can verify symbolically that the solution of $\Phi_{N/2}^\prime(\overline{k}_{N/2})=0$ is attained at $\overline{k}_{N/2}=-\sqrt{\sigma}$, a common stationary point of the family of function $\{\Phi_p\}_{p=1}^{N-1}$. As we are able to retrieve a closed-form expression of $\Phi_{N/2}(-\sqrt{\sigma})$, and by combining condition~\eqref{eq:condition_multiplicity}, we obtain the claim, i.e., 
\begin{equation*}
    \Phi_{N/2}(-\sqrt{\sigma})=-a+\sqrt{\sigma}(N-1)+\sqrt{\sigma+\frac1\gamma}\leq0.
\end{equation*}

If $N$ is odd, then
$\Phi_{\lfloor N/2\rfloor}(\bar{k}_{\lfloor N/2\rfloor})<\Phi_{\lfloor N/2\rfloor}(-\sqrt{\sigma})$,
so the same condition remains sufficient but is no longer necessary. We remark that in order to obtain a necessary condition for case of $N$ being odd, one should solve inequality~\eqref{eq:condition_multiplicity}, which cannot be done in closed-form. 
\end{proof}
A pseudo-code for computing all hyperbolic FNE is provided in Algorithm~\ref{alg:symm_FNE}. The search is based on the iterative verification of Condition~\eqref{eq:condition_multiplicity}. If the latter is not valid, then the loop is stopped as additional solutions cannot exist. Whenever the algorithm is called, the symmetric stable FNE is retrieved by simply substituting the game parameters inside the closed-form solution \eqref{eq:symm_FNE}.

\begin{algorithm}
\caption{FNE-seeking method - Symmetric LQ Game}
\label{alg:symm_FNE}
\begin{algorithmic}
\State \textbf{Input: } $a, \gamma, \sigma$
\State \textbf{Output:}  List of FNE
\While{condition~\eqref{eq:condition_multiplicity} holds and $p\leq\lfloor N/2\rfloor$}
\State Substitute $a, \gamma, \sigma$ into the expression $\Phi_p(k)=0$~\eqref{eq:phi_k}
\State Find a couple of roots $\left(k^\star_{1,p}, k^\star_{2,p}\right)$ 
\State Retrieve new batch of FNE
\EndWhile 
\end{algorithmic}
\end{algorithm}

\begin{figure}[b!]
    \centering
    \includegraphics[width=\linewidth]{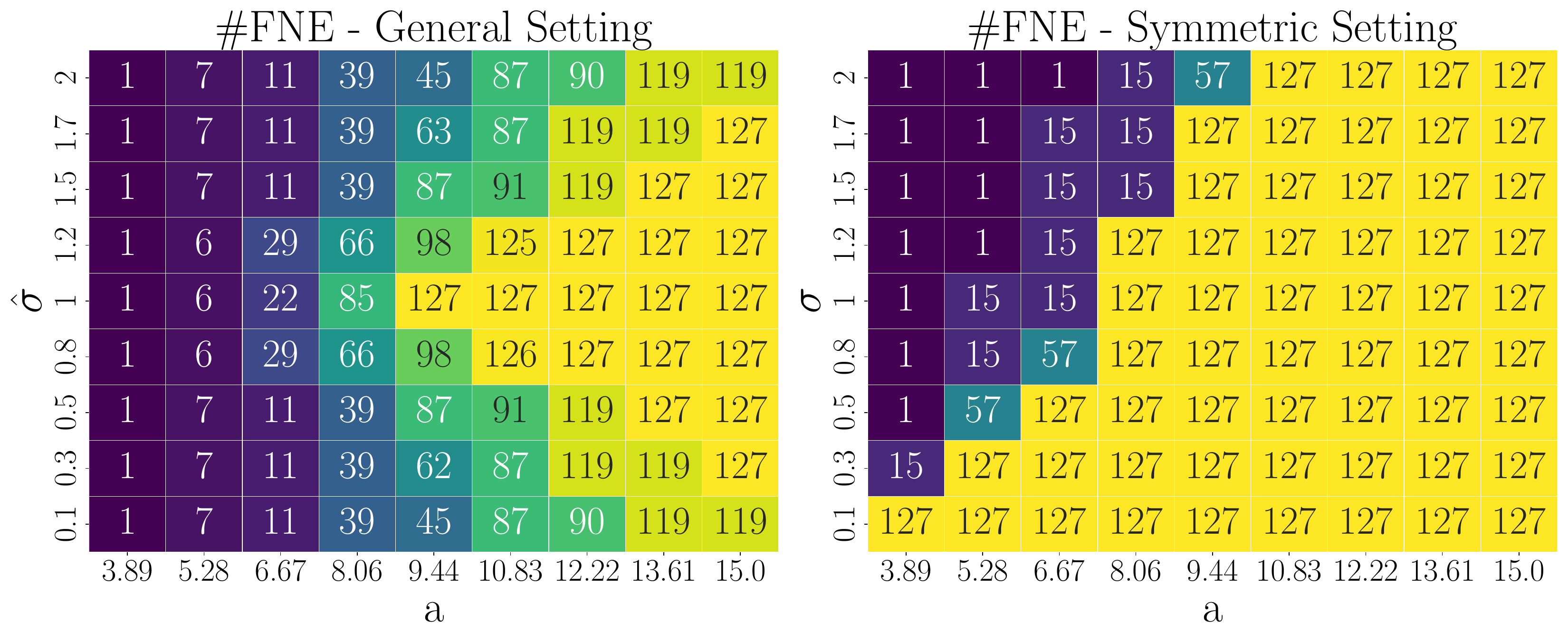}
    \caption{Number of stable FNE of different $7$-players game configurations computed by Algorithm \ref{alg:code} and Algorithm \ref{alg:symm_FNE}. As shown in each cell, the number of equilibria ranges from $1$ to a maximum of $2^N-1=127$.}
    \label{fig:n_FNE}
\end{figure}

\section{Numerical Results}\label{sec:numerical_results}
We test the algorithms to visualize the number of FNE arising across different game configurations. For a fixed game with $N=7$ players, we consider cost vectors of the form $(\hat{\sigma}, \hat{\sigma}, \hat{\sigma}, 0.5, \tilde{\sigma},\tilde{\sigma}, \tilde{\sigma})$, where $\hat{\sigma},\tilde{\sigma}\in\mathcal{T}=\{0.1,0.3,0.5,0.8,1,1.2,1.5,1.7,2\}$ and for every $\hat{\sigma}$ as the $j$-th element of $\mathcal{T}$, we set $\tilde{\sigma}$ as the $(\left|\mathcal{T}\right|-j)$-th element so as to highlight the effect of cost heterogeneity. Figure~\ref{fig:n_FNE} shows that equilibrium multiplicity is largest when players have similar costs, while it decreases as the game becomes more heterogeneous and some players become more conservative (larger $\sigma_i$). It is important to note that Algorithm~\ref{alg:code} has exponential complexity $\mathcal{O}(2^N)$, making it practical only for moderate values of $N$, while still extending beyond the small-scale settings discussed in~\cite{Salizzoni_2025_scalar}. By contrast, Algorithm~\ref{alg:symm_FNE} has lower complexity $\mathcal{O}(\lfloor N/2\rfloor)$, but applies only to symmetric games. We then investigate the role of the discount factor by fixing $N=7$, $a=5$, and cost vector $(0.1,0.1,0.1,0.15,0.2,0.2,0.2)$, for which condition~\eqref{eq:suff_cond} is violated when $\gamma\in[0.1,0.5]$. As shown in Figure~\ref{fig:gammas}, low discount factors yield finite-cost but non-stabilizing equilibria, while increasing $\gamma$ makes the number of stable FNE and FNE coincide. These numerical results highlight the importance of keeping the discount factor explicit and motivate the distinction between stable FNE and FNE when condition~\eqref{eq:suff_cond} is violated.
\section{Conclusions}
For the class of scalar infinite-horizon discounted $N$-player games, we distinguish FNE from stable FNE, showing that discounting may generate finite-cost equilibria that do not stabilize the system. We provide a sufficient condition ensuring that all FNE are stabilizing, derive a parametric characterization of all equilibria together with a numerical method for their computation, and specialize the analysis to the symmetric setting, where a closed-form symmetric FNE and multiplicity conditions are obtained. Numerical results illustrate how equilibrium multiplicity depends on the game parameters and how violating the stability condition may lead to finite-cost non-stabilizing equilibria.
\section*{Acknowledgment}
The authors would like to thank Prof. Tatarenko for valuable discussions during the development of this work.

\begin{figure}
    \centering
    \includegraphics[width=\linewidth]{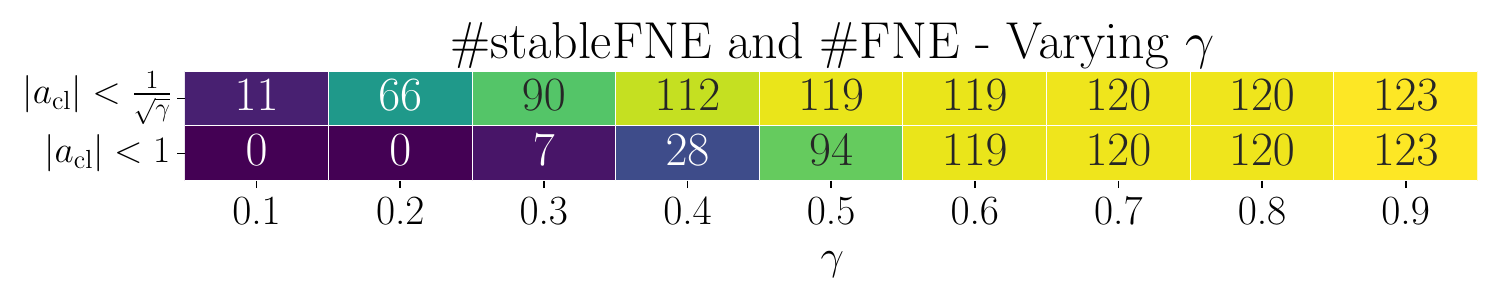}
    \caption{Number of stable FNE and FNE for the discounted LQ game. For a fixed configuration of $(N, a, \sigma_i)$, we vary $\gamma$ such that for $\gamma\in[0.1, 0.5]$ condition~\eqref{eq:suff_cond} is not satisfied, leading to equilibria with finite cost that do not stabilize the system. It is worth noting that the violation of~\eqref{eq:suff_cond} does not generally imply a discrepancy between the number of stable FNE and FNE since the condition is only sufficient. }
    \label{fig:gammas}
\end{figure}

\appendix
\subsection{Proof of Theorem~\ref{thm:symm}}\label{appendix:thm_symm}
Let us consider the reduced optimal system~\eqref{eq:optiaml_sys_symm}, and exploit symmetry, i.e., $k_i=k\in\mathbb{R}$, which entails $S=Nk$. Let us substitute the latter in $h^-(k)+k=S$ and multiply it by $2\gamma k$, which is non-zero by assumption. Let us isolate the square root on one side of the equation
\begin{equation}\label{eq:pre_cubic}
     \gamma \left[k^2\left(1-2N\right)- 2 a  k + \sigma \right]= \sqrt{\gamma^2\left(k^2+\sigma\right)^2+4\gamma k^2},
\end{equation}
with the aim to square both sides. Before doing so, we must take into account the positive domain of the left side of the expression, as squaring will include non-stable solutions. The expression $ \left[k^2\left(1-2N\right)- 2 a  k + \sigma \right]$ is positive if
\begin{equation}\label{eq:dom_cubic}
    k\in\left[ \frac{-a-\Delta}{2N-1},  \frac{-a+\Delta}{2N-1}\right],
\end{equation}
with $\Delta=\sqrt{2N\sigma+a^2-\sigma}$. In order to prove existence and uniqueness, one shall prove the thesis only for those $x$ satisfying~\eqref{eq:dom_cubic}. However, as one can easily verify from the symmetric equation $h^-(k)+k=Nk$, the intersection between the curves $h^-(k)+k$ and $Nk$ is always negative when $a$ is positive and vice-versa. Moreover, because $\frac{-a+\Delta}{2N-1}>0$ if and only if $N\geq\frac{1}{2}$, we can further refine~\eqref{eq:dom_cubic} as
\begin{equation}\label{eq:dom_cubic_red}
\left[\frac{-a-\Delta}{2N-1},0\right) \;\text{if } a>0,
\qquad
\left(0,\frac{-a+\Delta}{2N-1}\right] \;\text{if } a<0 .
\end{equation}
For simplicity, we will assume $a>0$ as the case $a<0$ follows analogously. Let us finally square the expression in~\eqref{eq:pre_cubic} by taking into account the correct domain~\eqref{eq:dom_cubic_red}
and obtain $N\gamma k^{3} \left(N - 1\right) + \gamma ak^{2} \left(2 N  - 1
\right) + k \left(- N \gamma \sigma + a^{2} \gamma - 1\right)- a \gamma \sigma=0$, i.e., a cubic that we denote as $\mathcal{C}(k)=0$. Let us compute the cubic at the right extremum of the domain~\eqref{eq:dom_cubic_red}, to obtain $\mathcal{C}\left(\frac{-a-\Delta}{2N-1}\right) = A\left[\left(2N-1\right)^2+\gamma\left(N\Delta-a(N-1)\right)^2\right]$, where $A=a+\Delta>0$. The expression is obtained after few algebraic substitutions, in particular the relation $\sigma=\frac{\Delta^2-a^2}{2N-1}$, but they are omitted for lack of space. Since the expression is a sum of squares, one can conclude that $\mathcal{C}\left(\frac{-a-\Delta}{2N-1}\right)>0$. On the other domain extremum, we obtain $\mathcal{C}(0)=-a\gamma \sigma$, which is negative, as we assumed $a>0$. Consequently, by the Intermediate Value Theorem, there exists a solution inside the domain of interest.
To prove uniqueness, one shall observe that $\mathcal{C}(k)$ is a cubic with positive leading term $N\gamma(N-1)>0$ which implies $\displaystyle\lim_{k\to\pm\infty}\mathcal{C}(k)=\pm\infty$. Moreover, we have shown that $\mathcal{C}(\overline{k})>0$, with $\overline{k}=\frac{-a-\Delta}{2N-1}<0$, such that, by continuity, the cubic must admit another real root in the negative half space. By the same reasoning, because $\mathcal{C}(0)<0$, another root exists in the positive half space. Consequently, three real distinct roots always exist, which proves uniqueness in the domain defined in~\eqref{eq:dom_cubic_red}. The expressions in~\eqref{eq:symm_FNE} are obtained through Cardano's procedures for the solution of cubic expressions.

\bibliographystyle{IEEEtran}  
\bibliography{bib} 
\end{document}